\documentclass[letterpaper,12pt]{article}
	\usepackage[utf8]{inputenc}
	\usepackage{amsfonts, amsmath, amssymb, amsthm}
	\usepackage{bbm}
	\usepackage{thmtools}
	\usepackage{algpseudocode}
	\usepackage{graphicx}
	\usepackage{stmaryrd}
	\usepackage[show]{ed}
	\usepackage[small]{complexity}
	\usepackage{ulem}
	\usepackage{xspace}
	\usepackage[margin=1in]{geometry}
	\usepackage[pagebackref=true]{hyperref}
	\usepackage{algorithm}

	\numberwithin{equation}{section}

	\theoremstyle{plain}
	\declaretheorem[numberlike=equation]{theorem}
	\declaretheorem[unnumbered,name=Theorem]{theorem*}
	\declaretheorem[numberlike=equation]{lemma}
	\declaretheorem[unnumbered,name=Lemma]{lemma*}
	
	\declaretheorem[unnumbered,name=Corollary]{corollary*}
	
	\declaretheorem[unnumbered,name=Proposition]{proposition*}
	
	\declaretheorem[unnumbered,name=Claim]{claim*}
	
	\declaretheorem[unnumbered,name=Conjecture]{conjecture*}

	\declaretheorem[numberlike=equation]{definition}
	\declaretheorem[unnumbered,name=Definition]{definition*}
	
	\declaretheorem[unnumbered,name=Example]{example*}
	
	\declaretheorem[unnumbered,name=Notation]{notation*}

	\theoremstyle{remark}
	
	\declaretheorem[unnumbered,name=Remark]{remark*}

	\DeclareMathOperator{\Enc}{Enc}

	\newcommand{\Z}{\mathbb{Z}}

	\newcommand{\cC}{\mathcal{C}}

	\newcommand{\eqdef}{:=}

\author{
	Michael Forbes
	\thanks{\texttt{miforbes@mit.edu}. Research was done while an intern at Microsoft Research, Silicon Valley.}
\and
	Sergey Yekhanin
	\thanks{\texttt{yekhanin@microsoft.com}}
}
\title{On the Locality of Codeword Symbols \\ in Non-Linear Codes}
\date{\today}

\begin{document}

\maketitle

\begin{abstract}
	Consider a possibly non-linear $(n,K,d)_q$ code. Coordinate $i$ has locality $r$ if its value is determined by some $r$ other coordinates. A recent line of work obtained an optimal trade-off between information locality of codes and their redundancy. Further, for linear codes meeting this trade-off, structure theorems were derived. In this work we give a new proof of the locality / redundancy trade-off and generalize structure theorems to non-linear codes.
\end{abstract}

\section{Introduction}

We say that a certain coordinate of an error-correcting code has locality $r$ if, when erased, the value at this coordinate can be recovered by accessing at most $r$ other coordinates. Motivated by applications to data storage~\cite{HuangSX+12} the authors of~\cite{2012-08-13.1} introduced $(r,d)$-codes, which are systematic codes that have distance $d$ and thus tolerate up to $d-1$ erasures, but also have the property that any information coordinate has locality $r$ or less. They established that in all linear $[n,k,d]_q$ codes with the $(r,d)$-property
\begin{equation}
	\label{Eqn:Fund}
	n\geq k+\left\lceil \frac{k}{r}\right\rceil+d-2.
\end{equation}
In what follows we refer to codes that meet~(\ref{Eqn:Fund}) with equality as \emph{optimal}. A construction of~\cite{HuangCL07} implies that optimal codes exist for all values of parameters. In the natural setting of $r|k$, the lower-bound argument of~\cite{2012-08-13.1} yields structure theorems for optimal linear codes. These theorems are particularly strong when $d<r+3.$ In particular, in that case they imply tight lower bounds for the locality of parity coordinates. The results of~\cite{2012-08-13.1} were recently extended by~\cite{PD_isit,XOR_ELE} who generalized the inequality~(\ref{Eqn:Fund}) (but not the structure theorems) to non-linear codes.

In this paper, we further extend the line of work above. We first give a new proof of the lower bound~(\ref{Eqn:Fund}) for non-linear codes. Then, as in~\cite{2012-08-13.1}, we use the lower-bound argument to derive structure theorems for optimal non-linear codes.

The main technical problem that we have to address going from the lower bound to structure theorems is that of reversibility of local constraints. In linear codes, any local constraint on coordinates in the code must be a linear constraint, and linear constraints are trivially reversible, in that knowing all but 1 coordinate in the constraint always determines that coordinate, regardless of the identity of that 1 coordinate.  However, for non-linear codes it is possible to have local constraints that are not reversible.  For example, it is possible for the coordinates $\{i,i'\}$ to determine the coordinate $i''$, but for the coordinates $\{i',i''\}$ to not determine the coordinate $i$.  However, we show that for optimal $(r,d)$-codes, even in the non-linear case, all locality constraints must be reversible.  Once this is established, the structural results of~\cite{2012-08-13.1} can then be extended to the non-linear case.

\section{Preliminaries}

We will first fix some notation, then define the objects we will be considering.

\subsection{Notation}

Throughout, we consider codes which may be non-linear over an arbitrary alphabet $\Sigma$, where $|\Sigma|=q\ge 2$ is an arbitrary integer. Given two vectors $\vec{x},\vec{y}\in\Sigma^n$, $\Delta(\vec{x},\vec{y})$ will denote the unnormalized Hamming distance $\vec{x}$ and $\vec{y}$. For $S\subseteq[n]$, we will denote $\vec{x}|_S$ for the sequence of symbols in $\vec{x}$ with coordinates in $S$.  When $S=\{i\}$ we will just write $\vec{x}|_i$.  For an integer $n\ge0$, $[n]$ denotes the set $\{1,\ldots,n\}$, where $[0]$ will be understood as the empty-set. For disjoint sets $A$ and $B$, we write $A\sqcup B$ to denote their disjoint union.

\subsection{Definitions}

Recall the definition of a code, which we do not assume to be linear.

\begin{definition}
	A $(n,K,d)_q$ code is a subset $\cC\subseteq\Sigma^n$ with size $|\cC|=K$, such for any $\vec{x}\ne\vec{y}\in\cC$, $\Delta(\vec{x},\vec{y})\ge d$. If $\cC'\subseteq \cC$ then $\cC'$ is a \textbf{sub-code} of $\cC$.  The parameter $n$ will referred to as the \textbf{block-length}, $k=\log_q K$ the \textbf{dimension} and $d$ the \textbf{distance}.
	
	The code is \textbf{systematic} if $k \in\Z$, and there is an encoding function $\Enc:\Sigma^k\to\Sigma^n$ such that for $\vec{x}\in\Sigma^k$, $\Enc(\vec{x})|_i=\vec{x}|_i$, for $i\in[k]$.
\end{definition}

A systematic code takes on all $q^k$ values in its first $k$ coordinates, and the values of these coordinates determine the rest of the codeword.  The first $k$ coordinates of the codewords are thus referred to as the information symbols, other coordinates will be called parity symbols. This work will be interested in codes with local constraints on the information symbols.

\begin{definition}
	A systematic $(n,K,d)_q$ code has \textbf{information locality r} if for every $i\in[k]$, there is a size $\le r$ subset $S\subseteq[n]\setminus\{i\}$ such that for any $\vec{x}\in\cC$, $\vec{x}|_i$ is determined by $\vec{x}|_S$.
\end{definition}

Other symbols, other than the information symbols, can also have locality, and occasionally we will use this.

\section{Locality Lower Bounds}

In this section we establish lower bounds on the block-length of codes with small information locality.  We will then prove structural results for codes meeting this lower bound. As we will often use it, we now prove the Singleton bound.

\begin{lemma}[Singleton Bound]
	\label{singleton}
	Let $\cC$ be an $(n,K,d)_q$ code, with $K>1$.  Then, \[n\ge\log_qK+d-1\]
\end{lemma}
\begin{proof}
	As $K>1$, there are at least two distinct codewords. These codewords are at least $d$-apart. Thus $d\le n$.  Therefore we can delete the first $d-1$ coordinates from each codeword, resulting in a code $\cC'\subseteq \Sigma^{n-d+1}$.  The new code $\cC'$ has distance $\ge 1$, as each pair of original codewords have distance $\ge d$ and we only deleted $d-1$ coordinates. Thus, we have an injective map from $\cC$ to $\Sigma^{n-d+1}$, and thus $\log_qK\le n-d+1$.
\end{proof}

The lower bounds we derive for local codes will follow by analyzing Algorithm~\ref{alg:subcode}. This algorithm will use the local constraints of the code to iteratively find large sub-codes $\cC_j\subseteq\cC\subseteq\Sigma^n$.  In this process, the (effective) block-length will decrease faster than the dimension of the sub-codes, while maintaining the distance.  Thus, the sub-codes become more optimal in terms of rate, and eventually we can apply the Singleton bound, to bound this process.

\begin{algorithm}\caption{Finding sub-codes via locality}
	\label{alg:subcode}
	\begin{algorithmic}[1]
		\Procedure{sub-code}{$\cC,n,k,d,r,q$}
			\State $\cC_0=\cC$
			\State $j=0$
			\While{$|\cC_{j}|>1$}
				\State $j\leftarrow j+1$.
				\State Choose $i_j$ such that $i_j\notin R_{j-1}\eqdef\bigcup_{j'\in[j-1]}(S_{j'}\cup\{i_{j'}\})$
				and the size $\le r$ subset of coordinates $S_j\subseteq[n]\setminus\{i_j\}$ determine the coordinate $i_j$, for all $\vec{x}\in\cC$.
				\label{alg:subcode:fixi}
				\State Let $\vec{\sigma}_j\in\Sigma^{|S_j|}$ be the most frequent element in the multi-set $\{\vec{x}|_{S_j}:\vec{x}\in\cC_{j-1}\}$.\label{alg:subcode:avg}
				\State Define $\cC_{j}\eqdef\{\vec{x}:\vec{x}\in\cC_{j-1}, \vec{x}|_{S_j}=\vec{\sigma}_j\}$.\label{alg:subcode:prune}
			\EndWhile
		\EndProcedure
	\end{algorithmic}
\end{algorithm}

\begin{theorem}
	\label{localsingleton}
	Let $\cC$ be a systematic $(n,q^k,d)_q$ code with information locality $r$.  Then \[n\ge k+\left\lceil\frac{k}{r}\right\rceil+d-2.\]
\end{theorem}
\begin{proof}
	We first show that Algorithm~\ref{alg:subcode} is well-defined.  In particular, in Line~\ref{alg:subcode:fixi}, such an $i_j$ exists.  For, by hypothesis $|\cC_{j-1}|>1$, implying there are $\vec{x}\ne\vec{y}\in\cC_{j-1}\subseteq\cC$.  As $\cC$ is systematic, such codewords are determined by their first $k$ coordinates, and thus must differ on at least one of those $k$ coordinates.  Further, they will not differ on coordinates in $R_{j-1}$, as those coordinates are fixed, as we fixed the coordinates in $S_{j'}$ for $j'<j$ by construction in Line~\ref{alg:subcode:prune} and this also fixes the $\{i_{j'}\}_{j'<j}$ by locality.  Thus, any coordinate where $\vec{x}$ and $\vec{y}$ differ will suffice for $i_j$.  In Line~\ref{alg:subcode:fixi} the set $S_j$ exists by our locality assumption on the code $\cC$. Thus, the algorithm is well-defined.
	
	We now analyze the algorithm. We first show that each new sub-code is not too small. Define $T_j\eqdef S_j\setminus R_{j-1}$ and define $t_j\eqdef |T_j|$, which is the number of coordinates we fix in constructing $\cC_j$ that are not necessarily fixed by prior loops in the procedure.  It follows that there are $\le q^{t_j}$ many possibilities for the $\vec{\sigma}_j$ in Line~\ref{alg:subcode:avg}, and thus $|\cC_j|\ge |\cC_{j-1}|/q^{t_j}$ by averaging.

	Given that the sub-codes are not shrinking too fast, we can use this to lower-bound the number of iterations of the algorithm. Let $\ell$ denote the largest value of $j$ such that $|\cC_j|>1$ in the algorithm, and thus $|\cC_{\ell+1}|=1$. By the above bound on sub-code size, we see then that,
	\[0=\log_q|\cC_{\ell+1}|\ge k-\sum_{j=1}^{\ell+1}t_j\]
	and so as $t_j\le |S_j|\le r$,
	\[k\le (\ell+1)r\]
	or equivalently,
	\[\ell\ge \left\lceil \frac{k}{r}\right\rceil -1.\]
	
	We now reduce to the Singleton bound. We first note that $R_\ell$ is the disjoint union $R_\ell=\bigsqcup_{j=1}^\ell (T_j\sqcup\{i_j\})$ and thus $|R_\ell|=\ell+\sum_{j=1}^\ell t_j$. This decomposition of $R_\ell$ is disjoint as the $T_j$ are disjoint by construction, the $i_j$ are distinct by construction, and the $i_j$ are disjoint from the $T_{j'}$: $i_j\notin S_j\supseteq T_j$ by definition of the $S_j$, $i_j\notin T_{j'}$ for $j'>j$ by definition of $T_{j'}$, and $i_j\notin R_{j'}\supseteq T_{j'}$ for $j'< j$ by definition of $i_j$.

	Now consider the code $\cC_\ell\subseteq\cC\subseteq\Sigma^n$, whose codewords all agree on $R_\ell$, by construction.  As it is a sub-code of $\cC$, $\cC_\ell$ also has minimum distance $\ge d$.  It follows that we can delete the coordinates $R_\ell\subseteq[n]$ and bijectively map $\cC_\ell$ to the new code $\cC'\subseteq\Sigma^{n-|R_\ell|}$, which still has distance $\ge d$. By the above arguments, we see that
	\[\log_q|\cC'|=\log_q|\cC_\ell|\ge\log_q|\cC|-\sum_{j=1}^\ell t_j=\log_q|\cC|-|R_\ell|+\ell\]
	Thus, applying the Singleton bound (Lemma~\ref{singleton}) to $\cC'$, we get that,
	\[n-|R_\ell|\ge k-|R_\ell|+\ell+d-1\]
	and thus
	\[n\ge k+\left\lceil \frac{k}{r}\right\rceil+d-2.\qedhere\]
\end{proof}

We now turn to structural results on optimal local codes, that is, codes with information locality $r$ that meet the above bound of $n=k+\lceil k/r\rceil +d-2$.  For simplicity, we will restrict ourselves to the case that $r|k$ so that $k/r$ is integral.  We will first show the local structure for $r=k$, in which case the code is of parameters $(k+d-1,q^k,d)_q$ and is thus a maximum-distance separable (MDS) code.

\begin{lemma}
	\label{MDSlocality}
	Let $\cC$ be a $(k+d-1,q^k,d)_q$ code, that is, a MDS code.  Then for any subset of coordinates $S\subseteq[k+d-1]$ of size $k$, the multi-set $\{\vec{x}|_S:\vec{x}\in\cC\}$ takes on all values in $\Sigma^k$ exactly once, and for any $\vec{x}\in\cC$, $\vec{x}|_S$ determines $\vec{x}$.
\end{lemma}
\begin{proof}
	Suppose $\vec{x},\vec{y}\in\cC$ (possibly equal) have $\vec{x}|_S=\vec{y}|_S$.  Then as $|S|=k$, it follows that $\Delta(\vec{x},\vec{y})\le (k+d-1)-k=d-1<d$. As the minimum distance of any two distinct codewords in $\cC$ is $\ge d$, it follows that $\vec{x}=\vec{y}$, so $\vec{x}|_S$ determines $\vec{x}$.

	Thus, as there are $q^k$ codewords, and the values of the coordinates $\vec{x}|_S$ (taking values in $\Sigma^k$, and $|\Sigma^k|=q^k$) determine the codeword, it follows that each value in $\Sigma^k$ is taken by $\vec{x}|_S$ exactly once, ranging over $\vec{x}\in\cC$.
\end{proof}

This lemma shows that for each symbol, MDS codes have locality $k$ and no less.  When $r|k$ but $r<k$ the situation becomes more complicated, and we derive the following result (Theorem~\ref{disjointlocality}). The heart of this result is item \eqref{disjointlocality:reversibility}, which establishes that in non-linear optimal $(r,d)$-codes, all local constraints will constrain the involved symbols equally, and are thus reversible.  This fact is trivial for linear codes, but more involved for non-linear codes.  Once established, the arguments can follow the case for linear codes as done in~\cite{2012-08-13.1}.

\begin{theorem}
	\label{disjointlocality}
	Let $\cC$ be a systematic $(n,q^k,d)_q$ code with information locality $r$, with $r|k$ and $r<k$.  Suppose $n=k+\frac{k}{r}+d-2$.  Let (possibly equal) information coordinates $i,i'\in[k]$, have associated subsets $S\subseteq[n]\setminus\{i\}$ and $S'\subseteq [n]\setminus\{i'\}$ of size $\le r$, such that $\vec{x}|_S$ determines $\vec{x}|_i$ and $\vec{x}|_{S'}$ determines $\vec{x}|_{i'}$, for all $\vec{x}\in\cC$.  Then
	\begin{enumerate}
		\item $|S|=r$.\label{disjointlocality:size}
		\item For all $i''\in S\cup \{i\}$, $\vec{x}|_{(S\cup\{i\})\setminus\{i''\}}$ determines $\vec{x}|_{i''}$, for all $\vec{x}\in\cC$.\label{disjointlocality:reversibility}
		\item $S\cup\{i\}$ and $S'\cup\{i'\}$ are either equal or disjoint.\label{disjointlocality:disjoint}
		\item Up to a permutation of coordinates, $\cC$ is a code with $k$ information symbols $I$, $k/r$ parities $L$, each depending on a disjoint set of $r$ information symbols, and $d-2$ other parties $H$, depending arbitrarily on the $k$ information symbols.\label{disjointlocality:iso}
	\end{enumerate}
\end{theorem}
\begin{proof}
	The proof will be by analyzing particular runs of Algorithm~\ref{alg:subcode}, using the analysis of Theorem~\ref{localsingleton}.  By showing that certain inequalities in that analysis must be tight, we will derive the desired results.

	We first establish further properties of the algorithm, in the case that $n=k+k/r+d-2$, by extending the analysis given in Theorem~\ref{localsingleton}.  In particular that analysis shows that
	\[n\ge k+\ell+d-1\]
	and
	\[\sum_{j=1}^{\ell+1}t_j\ge k\]
	where $t_j\le |S_j|\le r$.  By the hypothesis that $n=k+k/r+d-2$, we have that $\ell\le k/r-1\in\Z$, from which it follows that
	\[\sum_{j=1}^{\ell+1}t_j\le (\ell+1)r\le k\]
	Combining the above two equations shows that these inequalities, and those inequalities used to derive them, must be met with equality.  In particular, we have that $t_j=|S_j|=r$ for all $j$, and $\ell=k/r-1$. From this, we can derive that the subsets $S_j\cup\{i_j\}$ have $r+1$ distinct elements, and the family of subsets $\{S_j\cup\{i_j\}\}$ are disjoint.
	
	Further, it follows that the inequalities used in the analysis of Theorem~\ref{localsingleton} to establish that ``$n\ge k+\ell+d-1$'' must also be tight.  In particular, $|\cC_j|=|\cC_{j-1}|/q^r$, for all $j$, implying that $|\cC_j|=q^{k-rj}$, for all $j$.  By the construction of $\cC_j$ in Lines~\ref{alg:subcode:avg}--\ref{alg:subcode:prune}, it follows by an averaging argument that the multi-set $\{\vec{x}|_{S_j}:\vec{x}\in\cC_{j-1}\}$ has $q^r$ distinct elements (the maximum possible), each appearing equally often. In particular, this implies that any choice of $\vec{\sigma}\in\Sigma^r$ in Line~\ref{alg:subcode:avg} is valid, for any choice of the $\{\vec{\sigma}_{j'}\}_{j'<j}$ in the prior iteration of the loops.  Further, once we have chosen these $k/r=\ell+1$ values $\vec{\sigma}_j$, there is a unique codeword in $\vec{x}\in\cC$ such that $\vec{x}|_{S_j}=\vec{\sigma}_j$ for all $j$, as $|\cC_{k/r}|=1$ by construction.  It follows then that the values in the $k$ coordinates $\sqcup_j S_j$ are completely independent, take on all $q^k$ possible values, and uniquely determine a codeword in $\cC$.
	
	We will now use these facts applied to particular runs of the algorithm.

	\uline{\eqref{disjointlocality:size}:} Consider Algorithm~\ref{alg:subcode} where we choose $i_1\leftarrow i$ and $S_1\leftarrow S$.  The choice of $i_1$ is valid, for we are free to choose any $i_1\in[n]$, as $R_1=\emptyset$. The choice of $S_1$ is also valid, as $S$ is a valid locality constraint on $i_1$. We then continue with Algorithm~\ref{alg:subcode}, to define the sets $S_j$ and codes $\cC_j$.  From the analysis above, it follows that $|S_j|=r$ for all $j$, in particular $|S|=|S_1|=r$.
	
	\uline{\eqref{disjointlocality:reversibility}:} For each $\vec{\tau}\in\Sigma^{r-1}$, consider the map $f_{\vec{\tau}}:\Sigma\rightarrow \Sigma$, such that $f_{\vec{\tau}}(\vec{x}|_{i''})=\vec{x}|_{i}$ for all $\vec{x}\in\cC$ such that $\vec{x}|_{S\setminus\{i''\}}=\vec{\tau}$.  This is well-defined, as the coordinates $S=(S\setminus\{i''\})\cup\{i''\}$ determine the coordinate $i$ by locality, and the coordinates in $S$ take on all $q^r$ values by the analysis above, so for each $\vec{\tau}$, $f_{\vec{\tau}}$ must be defined for each input.  The claim will be established by showing that, for each $\vec{\tau}$, this map is in fact a bijection, which by will follow from showing that it is injective.  To do this, we will analyze properties of sub-codes of $\cC$, which will follow from analysis of Algorithm~\ref{alg:subcode}.
	
	As in \eqref{disjointlocality:size}, we first run Algorithm~\ref{alg:subcode} with $i_1\leftarrow i$ and $S_1\leftarrow S$, and the algorithm yields the $k/r$ coordinates $\{i_j\}_{j\in[k/r]}$ and respective $k/r$ locality constraints $\{S_j\}_{j\in[k/r]}$, such that the family of $(r+1)$-sized subsets $\{S_j\cup\{i_j\}\}_j$ are disjoint.
	
	Now observe that choosing the $\{S_j\cup\{i_j\}\}_j$ in reverse order is also a valid run of Algorithm~\ref{alg:subcode}, because the conditions in Line~\ref{alg:subcode:fixi} hold regardless of the order of $j$ in the sequence $(S_j\cup\{i_j\})_j$. Consider the code $\cC'$ resulting from the algorithm run in reverse order, where we have fixed the coordinates $\{S_j\cup\{i_j\}\}_{j>1}$, and so $\cC'$ is the last code in this reverse-run algorithm that has more than one codeword.  As such, $|\cC'|=q^r$, by the analysis above.  It has $n-(k/r-1)(r+1)=r+d-1$ non-fixed coordinates, and distance $\ge d$.  The non-fixed coordinates include $i$, $S$, and the remaining $d-2$ coordinates~\footnote{Note that $d\ge2$ is implied here.  That $d=1$ is impossible can be seen most easily in the case when $r=k$.  For then, the parameters imply that the code is simply the list of all $q^k$ words, and so no information locality is possible, for each coordinate is fully independent of the rest.  That $d=1$ is impossible for $r<k$ can be seen by reducing to a sub-code where $r=k$ via Algorithm~\ref{alg:subcode}.}. This implies that $\cC'$ (when the fixed coordinates are dropped) is an MDS code. In particular, consider the two sets of coordinates, $S$ and $S\cup\{i\}\setminus\{i''\}$, in the code $\cC'$.  By Lemma~\ref{MDSlocality}, both of these sets of  coordinates take on all possible values in $\Sigma^r$, and determine the other $d-1$ coordinates in $\cC'$.
	
	We now show $f_{\vec{\tau}}$ is injective, for any $\vec{\tau}\in\Sigma^{r-1}$.  Suppose $f_{\vec{\tau}}(\rho)=f_{\vec{\tau}}(\omega)=\sigma$, for some (possibly equal) $\rho,\omega,\sigma\in\Sigma$. As the coordinates $S$ take on all $q^r$ values in $\cC'$, there are two codewords, $\vec{x},\vec{y}\in\cC'$ such that $\vec{x}|_{S\setminus\{i''\}}=\vec{y}|_{S\setminus\{i''\}}=\vec{\tau}$, $\vec{x}|_{i''}=\rho$ and $\vec{y}|_{i''}=\omega$.  By the locality in $\cC$, this means that  $\vec{x}|_{i}=\vec{y}|_{i}=\sigma$.  By the locality in $\cC'$ just established, $\vec{x}|_{i''}$ is determined by $\vec{x}|_{S\setminus\{i''\}}=\vec{y}|_{S\setminus\{i''\}}=\vec{\tau}$ and $\vec{x}|_{i}=\vec{y}|_{i}=\sigma$, and thus $\rho=\vec{x}|_{i''}=\vec{y}|_{i''}=\omega$. Thus, $f_{\vec{\tau}}$ must be injective, for every $\vec{\tau}$.

	\uline{\eqref{disjointlocality:disjoint}:} Suppose $S\cup\{i\}$ and $S'\cup\{i'\}$ are not equal, and we seek to show they are disjoint.  Let $i''$ be the index of a coordinate in one of the sets but not the other, and without loss of generality, $i''\in(S\cup\{i\})\setminus (S'\cup\{i'\})$.  Define $S''\eqdef S\cup\{i\}\setminus\{i''\}$, and observe that by \eqref{disjointlocality:reversibility} applied to $i''\in S\cup\{i\}$ implies that the coordinates in $S''$ determine the coordinate $i''$.

	Now run Algorithm~\ref{alg:subcode}, choosing $i_1\leftarrow i'$, $S_1\leftarrow S'$, and $i_2\leftarrow i''$, $S_2\leftarrow S''$.  The first round choice of coordinate/locality is clearly valid.  That a second round will even by executed follows from the above analysis, as the code $\cC_1$ has size $q^{k-r}>1$, using that $r<k$.  Further, the choices of coordinate/locality are valid because $i''\notin R_1=S'\cup\{i'\}$ and $S''$ determines $i''$ as established above.  Thus, we have a valid initial run of the algorithm, which we can than run to completion to yield the family $\{S_j\cup\{i_j\}\}_j$.  By the above analysis of the algorithm, it follows that the sets $\{S_j\cup\{i_j\}\}_j$ are disjoint.  In particular, $S'\cup\{i'\}$ and $S''\cup\{i''\}=S\cup\{i\}$ are disjoint, as desired.

	\uline{\eqref{disjointlocality:iso}:} The above analysis shows that a run of Algorithm~\ref{alg:subcode} yields the $k/r$ disjoint sets of coordinates $\{S_j\}_j$ that take on all $q^k$ values and uniquely determine the codeword. Treating these $k$ coordinates as information symbols, and the corresponding coordinates $\{i_j\}_j$ as the $k/r$ parities $L$, we get the desired reordering of the coordinates.
\end{proof}

The above analysis shows that the locality constraints are disjoint sets of size $r+1$, but does not indicate how many such constraints there are.  In the case that $d<r+3$, we can show that there are exactly $k/r$ such local constraints, and can characterize their structure.  In fact, just as in the linear case in~\cite{2012-08-13.1}, we see that the locality structure of optimal $(r,d)$-codes resembles that of Pyramid codes~\cite{HuangCL07}.

\begin{theorem}
	\label{canonicalstructure}
	Let $\cC$ be a systematic $(n,q^k,d)_q$ code with information locality $r$, with $r|k$ and $r<k$.  Suppose $n=k+\frac{k}{r}+d-2$ and $d<r+3$. Then the $k/r+d-2$ parity symbols can be partitioned into $L$ and $H$, with $|L|=k/r$ and $|H|=d-2$, where
	\begin{enumerate}
		\item The parities in $L$, each depend on a disjoint subset of size $r$ of the $k$ information symbols.\label{canonicalstructure:lightparities}
		\item The parities in $H$, each depend on all of the $k$ information symbols.\label{canonicalstructure:heavyparities}
		\item The parities in $L$ have locality exactly $r$.\label{canonicalstructure:lightlocalitylb}
		\item The parities in $H$ have locality $\ge k-(k/r-1)(d-3)>r$.\label{canonicalstructure:heavylocalitylb}
	\end{enumerate}
\end{theorem}
\begin{proof}
	\uline{\eqref{canonicalstructure:lightparities}:} We continue with the analysis given in the proof of Theorem~\ref{disjointlocality}.  In particular, it implies that the $n$ coordinates carry a family of subsets of size $r+1$, such that any $r$ of the $r+1$ coordinates determine the other.  By hypothesis, each information coordinate participates in such a subset, and the analysis in Theorem~\ref{disjointlocality} shows that there are at least $k/r$ such subsets.  This leaves at most $n-k/r\cdot(r+1)=d-2$ coordinates that do not participate in these locality constraints, and as $d-2<r+1$, there are in fact exactly $k/r$ disjoint locality constraints and exactly $d-2$ coordinates not covered by locality constraints.  Thus, there are $k/r\cdot(r+1)-k=k/r$ parity symbols participating in locality constraints, and let this set be $L$, which has the desired properties.

	\uline{\eqref{canonicalstructure:heavyparities}:} We now show that the parities in $H\eqdef [n]\setminus ([k]\cup L)$ depend on each of the information symbols.  Consider some information symbol $i\in [k]$, and consider $\sigma\ne\sigma'\in\Sigma$.  Let $\vec{x},\vec{x}'\in\Sigma^k$ be such that $\vec{x}|_i=\sigma$ and $\vec{x}'|_i=\sigma'$, and $\vec{x}|_j=\vec{x}'|_j$ for $i\ne j\in[k]$.  It follows then that $\Enc(\vec{x})\ne\Enc(\vec{x}')$, and these codewords agree in $(k-1)+k/r-1$ places: they agree in $k-1$ information symbols by construction, and they agree in all but one of the $k/r$ light parities (the light parity grouped with coordinate $i$ being the exception).  As the code has minimum distance $d$, and there are only $d$ coordinates that these distinct codewords can differ on, it follows that $\Enc(\vec{x})$ and $\Enc(\vec{x}')$ differ on all these coordinates, in particular the $d-2$ heavy parities $H$.  Thus, changing any information coordinate will change all heavy parities, showing that each coordinate in $H$ depends on each of the $k$ information coordinates.

	\uline{\eqref{canonicalstructure:lightlocalitylb}:} This follows from Theorem~\ref{disjointlocality}, as the light parities $L$ can, under a permutation of coordinates, be regarded as information symbols, and thus cannot have locality $<r$.

	\uline{\eqref{canonicalstructure:heavylocalitylb}:} Let $[k]=\sqcup_{j\in[k/r]} I_j$ be the partition of the information symbols into size $r$ subsets based on the grouping defined by the light parities. Pick arbitrary $\vec{\sigma}_j\in\Sigma^r$ for $j\in[k/r]$.  Define the code $\cC_j\subseteq\cC$ to be $\cC_j\eqdef\{\vec{x}:\vec{x}\in\cC, \vec{x}|_{I_j'}=\vec{\sigma}_{j'}, j'\in[k/r]\setminus\{j\}\}$.  It follows that in $\cC_j$, all light parities are fixed except for the $j$-th.  Thus, $\cC_j$ has $(r+1)+(d-2)=r+d-1$ unfixed symbols and has $q^r$ codewords, as all but $r$ information symbols are fixed.  As $\cC_j$ is a sub-code of $\cC$, it follows that it has distance $\ge d$ and is thus is an MDS code.

	Consider any heavy parity $h\in H$ determined by a set of coordinates $S\subseteq[n]\setminus\{h\}$ for all codewords in $\cC$, and thus for all codewords in $\cC_j$ for any $j$. By Lemma~\ref{MDSlocality}, we see that any $r$ symbols in $\cC_j$ are independent, so any locality constraint for $h$ in $\cC_j$ must involve at least $r$ other symbols.  As codewords in $\cC_j$ are only unfixed on the coordinates $I_j$ and $H$, it follows that $|S\cap(I_j\cup H\setminus\{h\})|\ge r$.  Ranging this inequality over all $j\in[k/r]$, and accounting for double counting over $H\setminus\{h\}$, we see that $|S|\ge k-(k/r-1)(d-3)$, as desired.
\end{proof}

\bibliographystyle{alphaurl}
\bibliography{code-locality}

\end{document}